\documentclass[twoside,11pt]{article}

\usepackage{template/jair} 
\usepackage{template/theapa}
\usepackage{rawfonts}

\usepackage{amsmath}
\usepackage{amsfonts}
\usepackage{amsthm}
\usepackage{amssymb}
\usepackage{tikz}
\usetikzlibrary{automata, positioning, arrows, shapes.geometric}

\usepackage{caption}
\usepackage{subcaption}

\usepackage{algorithm}
\usepackage{algpseudocode} 

\usepackage{acro}
\DeclareAcronym{sw}{
	short = SW, long = Sure Winning ,
	tag = abbrev
}

\DeclareAcronym{asw}{
	short = ASW, long = Almost-Sure Winning ,
	tag = abbrev
}

\DeclareAcronym{lsw}{
	short = LSW, long = Limit-Sure Winning ,
	tag = abbrev
}

\DeclareAcronym{dasw}{
	short = DASW, long = Deceptive Almost-Sure Winning ,
	tag = abbrev
}

\DeclareAcronym{pps}{
	short = PPS, long = Perceptually Permissive Strategy ,
	tag = abbrev
}

\DeclareAcronym{rpps}{
	short = RPPS, long = Randomized Perceptually Permissive Strategy ,
	tag = abbrev
}

\usepackage{booktabs}
\usepackage{tabularx}
\newcolumntype{C}{>{\centering\arraybackslash}X}

\usepackage{todonotes}

\newif\ifuseboldmathops
\newif\ifuseittextabbrevs
\useboldmathopstrue   
\useittextabbrevstrue 

\ifuseittextabbrevs

	\newcommand{\ie}{{\it i.e.~}}

\else

	\newcommand{\ie}{i.e.~}

\fi

\ifuseboldmathops

\else

\fi

\ifuseboldmathops

\else

\fi

\ifuseboldmathops

\else

\fi

\ifuseboldmathops

\else

\fi

\newcommand{\Eventually}{\Diamond \, }

\newcommand{\until}{\mbox{$\, {\sf U}\,$}}

\newcommand{\dist}[1]{\mathcal{D}(#1)}
\newcommand{\supp}{\mathsf{Supp}}

\newtheorem{theorem}{Theorem}
\newtheorem{definition}{Definition}
\newtheorem{example}{Example}
\newtheorem{problem}{Problem}
\newtheorem{lemma}{Lemma}
\newtheorem{assumption}{Assumption}

\newcommand{\win}{\mathsf{Win}}
\newcommand{\occ}{\mathsf{Occ}}
\newcommand{\game}{\mathcal{G}}
\newcommand{\hgame}{\mathcal{H}}

\theoremstyle{plain}
\ifdefined \axiom \else
    \newtheorem{axiom}{Axiom}
\fi

\ifdefined \claim \else
    \newtheorem{claim}{Claim}
\fi

\newtheorem{proposition}{Proposition}
\newtheorem{corollary}{Corollary}[theorem]

\newtheorem{propcorollary}{Corollary}[proposition]

\ifdefined \lemma \else
  \newtheorem{lemma}{Lemma}
\fi

\ifdefined \theorem \else
  \newtheorem{theorem}{Theorem}
\fi

\ifdefined \question \else
    \newtheorem{question}{Question}
\fi

\theoremstyle{definition}
\newtheorem{notation}{Notation}

\ifdefined \problem \else
  \newtheorem{problem}{Problem}
\fi

\ifdefined \definition \else
  \newtheorem{definition}{Definition}
\fi

\ifdefined \assumption \else
  \newtheorem{assumption}{Assumption}
\fi

\ifdefined \example \else
  \newtheorem{example}{Example}
\fi

\theoremstyle{remark}

\newcommand{\fig}[1]{\mbox{Fig.~\ref{#1}}}

\newcommand{\lma}[1]{Lemma~\ref{#1}}

\newcommand{\ex}[1]{Example~\ref{#1}}

\DeclareAcronym{ltl}{
	short = LTL, long = Linear Temporal Logic ,
	tag = abbrev
}

\DeclareAcronym{scltl}{
	short = sc-LTL, long = Syntactically Co-safe LTL ,
	tag = abbrev
}

\DeclareAcronym{dfa}{
	short = DFA, long = Deterministic Finite Automaton ,
	tag = abbrev
}

\DeclareAcronym{mdp}{
	short = MDP, long = Markov Decision Process ,
	tag = abbrev
}

\newcommand{\prob}{Pr}

\newcommand{\act}{{A}}

\newcommand{\dawin}{{\mathsf{DAWin}}}
\newcommand{\dswin}{{\mathsf{DSWin}}}
\newcommand{\dapre}{{\mathsf{DAPre}}}
\newcommand{\calAP}{{\mathcal{AP}}}

\newcommand{\calF}{{\mathcal{F}}}

\usepackage{soul}
\usepackage{setspace}
\ShortHeadings{Synthesis with Action Deception}
{Kulkarni, \& Fu}
\firstpageno{1}

\begin{document}

\title{Synthesis of Deceptive Strategies in Reachability Games with Action Misperception\\ \small (Technical Report)}

\author{\name Abhishek N. Kulkarni \email ankulkarni@wpi.edu \\
       \addr Worcester Polytechnic Institute, 100 Institute Road, \\
       Worcester, MA 01609 USA
       \AND
       \name Jie Fu \email jfu2@wpi.edu \\
       \addr Worcester Polytechnic Institute, 100 Institute Road, \\
       Worcester, MA 01609 USA
       }


\maketitle

\begin{abstract}
Strategic deception is an act of manipulating the opponent's perception to gain strategic advantages. In this paper, we study synthesis of deceptive winning strategies in two-player turn-based zero-sum reachability games on graphs with one-sided incomplete information of action sets. In particular, we consider the class of games in which Player 1 (P1) starts with a non-empty set of private actions, which she may `reveal' to Player 2 (P2) during the course of the game. P2 is equipped with an \emph{inference mechanism} using which he updates his perception of P1's action set whenever a new action is revealed. Under this information structure, the objective of P1 is to reach a set of goal states in the game graph while that of P2 is to prevent it. We address the question: \emph{how can P1 leverage her information advantages to deceive P2 into choosing actions that in turn benefit P1?} To this end, we introduce a dynamic hypergame model to capture the reachability game with evolving misperception of P2. Analyzing the game qualitatively, we design algorithms to synthesize deceptive sure and almost-sure winning regions, and establish two key results: (1) under sure-winning condition, deceptive winning strategy is equivalent to the non-deceptive winning strategy---\ie use of deception has no advantages, (2) under almost-sure winning condition, the deceptive winning strategy could be more powerful than the non-deceptive strategy. We illustrate our algorithms using a capture-the-flag game, and demonstrate the use of proposed approach to a larger class of games with temporal logic objectives.

\end{abstract}


\section{Introduction}
    \label{sec:introduction}

In a two-player reachability game, a controllable player P1 (player 1, pronoun `she') plays against an uncontrollable player P2 (player 2, pronoun `he') to reach a set of goal  states (also called final states for P1). Synthesis of winning strategies in reachability games on (finite) graphs is a central problem in several areas such as model checking \cite{clarke2018model,baier2008principles}, reactive synthesis \cite{pnueli1989synthesis}, control of discrete event systems \cite{ramadge1989control}, robotics \cite{fainekos2009temporal} and cybersecurity \cite{jha2002two,aslanyan2016quantitative}. The solutions to reachability games often also provide the basis for solving more complex $\omega$-regular games \cite{gradel2002automata}. In literature, two-player reachability games  have been extensively studied for the case in which both the players have symmetric and complete information  \cite{mcnaughton1993infinite,zielonka1998infinite,gradel2002automata,gradel2002automata,chatterjee2012survey}. However, the solution concepts for such games under asymmetric incomplete information have not been thoroughly studied. The asymmetric incomplete information in games means at least one players have incomplete 
knowledge about some of the game construct: states, actions, transition functions, or goal states/payoffs.  

In this paper, we address the problem of synthesizing winning strategies in two-player, deterministic, turn-based reachability games with one-sided incomplete information. Specifically, as the game starts, P1 has a set of private actions which are unknown to P2. Additionally, P1 knows P2's action set and also knows that P2 does not know P1's private actions. During the course of the game, P1 may use any of her private actions. Due to perfect observation, such private actions will be revealed to P2. We equip P2 with an \emph{inference mechanism} using which he may update his perception of P1's action set whenever P1 reveals a private action. Such a game of asymmetric incomplete information has been investigated for normal-form games \cite{rasmusen1989games}. Asymmetrical information between players are commonly encountered in conflict analysis \cite{hipel2020graph}, cybersecurity \cite{shiva2010game,carroll2011game,zhuang2010modeling,hespanha2000deception}, auctions \cite{brandt2003fundamental}, and decision making for autonomous systems \cite{thing2016autonomous}. 
In such games, we are interested to know \emph{whether P1 can leverage her information advantages to deceive P2 into choosing actions that in turn benefit P1?} 


We approach the above question by modeling the interaction between P1 and P2 as a hypergame. A hypergame, first introduced in \cite{bennett1977toward}, models an interaction between two players in which they may have incomplete information about their opponent's action capabilities, strategies, preferences or objectives \cite{wang1989solution}. Hypergames can have many levels of perception because one player may have misperceptions about the opponent's interpretations of their interaction. Thus, instead of using a single game to model this situation, a hypergame represents a set of \emph{perceptual games} that capture the interaction as perceived by the players, given the information known to them. 
However, most of the solution concepts studied in the literature for hypergames assume that the (mis)perceptions of the players do not change during the interaction \cite{wang1989solution,sasaki2014subjective} or focuses on payoff deception \cite{gharesifard2013stealthy} where one player synthesizes a stealthy strategy to hide the private information of payoff functions in normal-form games.

In this paper, we introduce a new hypergame model, called a \emph{dynamic hypergame}, which allows the perception of players to evolve during the game. Specifically, when P1 reveals a private action, P2 updates his perception of P1's action set and, thereby, his counter-strategy. We propose two algorithms for \emph{qualitative} analysis of dynamic hypergames: Algorithm~\ref{alg:dsw} to compute the \emph{deceptive sure winning region}, \ie the set of states from which P1 has a strategy to reach the final states in \emph{finitely many} steps by strategically revealing the private actions, and Algorithm~\ref{alg:DASW} to compute the \emph{deceptive almost-sure winning region}, \ie the set of states from which P1 has a strategy to reach the final states with \emph{probability one and a undetermined number of steps} by strategically revealing the private actions. We note that by strategically revealing her private actions, P1 consciously controls P2's perception to her advantage, which is a deceptive behavior \cite{ettinger2010theory}. We assess the advantage of deception by comparing the size of deceptive sure and almost-sure winning regions computed using dynamic hypergames with the respective winning regions in the corresponding game with complete, symmetric information. In particular, we say deception is advantageous when there exists a state which is sure (resp., almost-sure) losing for P1 in a game with complete, symmetric information, but is \emph{deceptively} sure (resp., almost-sure) winning for her in the game with one-sided incomplete information.

Based on two proposed algorithms, we derive two important results for this class of games with one-sided incomplete information: (i) under the sure winning condition, P1 gains \emph{no advantage} by using deception, and (ii) under the almost-sure winning condition, in some game configurations, P1  gains advantage by using deception. Specifically, P1 can ensure to achieve the reachability objective with probability one by initially misinforming P2 of his action capabilities.

\subsection{Related Work}

In the games with incomplete information, both players have perfect observability but at least one of them has \emph{incomplete} information about at least one of the following components  \cite{levin2002games}: (a) the action capabilities of the opponent, (b) the objectives of the opponent, (c) the game rules, and (d) what knowledge does the opponent have about what I know about his knowledge about ... \textit{ad infinitum}. This class of games differs from games with imperfect information, in which players have complete knowledge of all aspects (a)-(d), but may not have perfect observation about the history (state-action sequences in the game) \cite{morgenstern1953theory}.

Games with incomplete information have been studied extensively using two models: Bayesian games and hypergames. Bayesian games, introduced by Harsanyi \cite{harsanyi1967games}, transform a game with incomplete information to a game with imperfect information by capturing players incomplete information as a type variable, which is not observable to other players. However, this transformation depends on the so-called \emph{consistency of priors} assumption which states that the set of possible types of players is a common knowledge.

Hypergames, first introduced in \cite{bennett1977toward}, do not impose the \emph{consistency of priors} assumption. As a result, both players can play different games, which they construct in their minds based on the information available to them. This property of hypergames allows us to explicitly model the \emph{unawareness} of the player, which is often exploited in deception. In addition, solving for winning strategies is also computationally less expensive than solving Bayesian games \cite{sasaki2012hypergames}. This is because the analysis of hypergames under subjective rationalizability requires us to consider players' behavior in only a few perceptual games. Whereas, we need to consider the best responses given all possible reachable beliefs of other's types to solve for Bayesian Nash equilibrium.

In the past, hypergame model has been used to study deception \cite{gutierrez2015modeling,kovach2016temporal,kovach2019trust}. These papers mainly focus on extending the notion of Nash equilibrium to level-$k$ normal form hypergames. The authors \cite{gharesifard2013stealthy} use the notion of H-digraph to establish necessary and sufficient conditions for deceivability. An H-digraph models a hypergame as a graph with nodes representing different outcomes in a normal-form game and edges representing a perceived improvement of outcome for a player. However, the game model studied in our paper is not a normal-form game, but instead a game on graph. A hypergame model based on a game on graph has been defined in our previous work \cite{kulkarni2020deceptive,li2020dynamic}, wherein we study the games with one-sided incomplete information about payoffs in temporal logic.

This paper is an extended version of our previous work \cite{kulkarni2020synthesis}, which studied the deceptive almost-sure winning with action deception. In comparison to that, we newly introduce an algorithm to compute deceptive strategies under sure-winning condition and provide the theoretical analysis on the advantages gained by using deception under sure and almost-sure winning conditions. We also include an experiment to synthesize deceptive winning strategy for P1 and show the extension of our solution approach from reachability games to games on graphs in which player's objectives are given as temporal logic formulas.


\paragraph{Structure of this paper.} After recalling the preliminaries in Section~\ref{sec:preliminaries}, we formalize the problem statement in Section~\ref{sec:problem-formulation}. Section~\ref{sec:dynamic-hypergame-on-graph} presents the dynamic hypergame model to capture the interaction between P1 and P2. Sections~\ref{sec:dsw-synthesis} and \ref{sec:dasw-synthesis} define the notions of deceptive sure and almost-sure winning strategies and present algorithms to synthesize them. The key results of this papers are derived in these sections. Lastly, Section~\ref{sec:experiment} demonstrates our approach using a gridworld motion planning problem given scLTL specifications. We conclude the paper in Section~\ref{sec:conclusion} by stating our conclusions and future work.









\section{Preliminaries}
    \label{sec:preliminaries}

\textbf{Notation.} 
Given a set $X$, we denote a probability distribution over $X$ by $d: X \rightarrow [0, 1]$, the set of all probability distributions over $X$ by $\dist{X}$ and the support of a distribution $d \in \dist{X}$ by $\supp(d) = \{ x \in X \mid d(x) > 0\}$.

In this section, we review the traditional approach used to analyze reachability games with complete, symmetric information. 
A reachability game models the interaction between two players P1 and P2 over a graph. P1's objective in such a game is to visit a final state and that of P2 is to prevent P1 from completing her task.

\begin{definition}[Reachability Game with Symmetric Information] 
    \label{def:game-on-graph}
    
    A two-player, deterministic, turn-based, zero-sum reachability game on graph with complete, symmetric information is a tuple \[\game = \langle S, A_1 \cup A_2, T, F \rangle ,\] where 
    
    \begin{itemize}
        \item $S = S_1 \cup S_2$ is the set of states partitioned into P1's states, $S_1$, and P2's states, $S_2$. P1 chooses an action when $s \in S_1$ and P2 chooses an action when $s \in S_2$;
        
        \item $A_1$ and $A_2$ are the set of actions of P1 and P2, respectively. The set of all actions is denoted by $\act = A_1 \cup A_2$;
        
        \item $T : S \times \act \rightarrow S$ is a deterministic transition function that maps a state and an action to a successor state;
        
        \item $F \subseteq S$ is a set of final states. 
    \end{itemize} 
\end{definition}

In a reachability game with complete, symmetric information, the game structure $\game$ is known to both players. Given an initial state $s_0 \in S$, a \textit{game-play} is constructed as an infinite sequence of state-action pairs selected by two players $\tau = s_0 a_0 s_1 a_1 s_2 a_2 \ldots$ such that $s_{i+1} = T(s_i, a_i)$ for all $i \geq 0$. A \textit{game-run} is the projection of the game-play onto the state set $S$ and is denoted by $\rho =  \tau \downharpoonright_S = s_0 s_1 s_2 \ldots$. Similarly, an \textit{action-history} is the projection of game-play onto the action set $A$, and is denoted by $\alpha = \tau \downharpoonright_\act = a_0 a_1 a_2 \ldots$. The $i$-th element of $\rho$ and $\alpha$ are denoted by $\rho_i$ and $\alpha_i$, respectively. Let $\occ(\rho) = \{s \in S \mid \exists k \in \mathbb{N} \text{ s.t. } s = \rho_k \}$ be the set of all states that appear in the game-run $\rho$. A game-play is said to be \emph{winning} for P1 if $\occ(\rho) \cap F \neq \emptyset$. Otherwise, the game-play is said to be \emph{winning} for P2.

A \emph{memoryless}, \emph{randomized} strategy of player $j \in \{1, 2\}$ is a function $\pi_j: S_j \rightarrow \dist{A_j}$. A strategy is said to be \emph{deterministic} if the support of $\pi_j(s)$ is singleton for all $s \in S_j$. The set of all memoryless strategies for player $j$ is denoted as $\Pi_j$. We only consider memoryless strategies because reachability games enjoy memoryless determinacy \cite{Mazala2002}. We refer to a pair of strategies $(\pi_1, \pi_2) \in \Pi_1 \times \Pi_2$ as a strategy profile. A run $\rho$ is said to be compatible with a strategy profile $(\pi_1, \pi_2)$ if for any state $s_i \in S_j$, there exists an action $a \in A_j$ such that $\pi_j(s_i)(a) > 0$ and $s_{i+1} = T(s_i, a)$. The set of all possible game-runs starting at $s \in S$ and compatible with the strategy profile $(\pi_1, \pi_2)$ is denoted by $\mathsf{Outcomes}(s, \pi_1, \pi_2)$. Given a strategy profile $(\pi_1, \pi_2)$, the probability that a state in $F$ is visited from some state $s \in S$ is denoted by $\prob_s^{\pi_1, \pi_2}(F) = \prob(\occ(\rho) \cap F \neq \emptyset \mid \rho \in \mathsf{Outcomes}(s, \pi_1, \pi_2))$.


Next, we recall the notions of \emph{sure} and \emph{almost-sure winning} in the reachability game with complete, symmetric information. 

\begin{definition}[Sure Winning Strategy]
    \label{def:sw-strategy}
    
    A memoryless strategy $\pi_1 \in \Pi_1$ is said to be \textit{sure winning} for P1 at a state $s \in S$ if and only if, for any strategy $\pi_2 \in \Pi_2$ of P2, every run $\rho \in \mathsf{Outcomes}(s, \pi_1, \pi_2)$ satisfies $\occ(\rho) \cap F \neq \emptyset$. 
\end{definition}

\begin{definition}[Almost-sure Winning Strategy]
    \label{def:asw-strategy}
    
    A memoryless strategy $\pi_1 \in \Pi_1$ is said to be \textit{almost-sure winning} for P1 at a state $s \in S$ if and only if, for any strategy $\pi_2 \in \Pi_2$ of P2, we have that $Pr_s^{\pi_1, \pi_2}(F) = 1$, that is, a state in $F$ is visited with probability one. 
\end{definition}

A game state is called a \textit{sure (almost-sure) winning state} for P1 if and only if P1 has a sure (almost-sure) winning strategy from that state. The set of all sure (almost-sure) winning states of P1 is called the \textit{sure (almost-sure)} \textit{region} of P1. The sure and almost-sure winning strategies, states and regions are defined for P2 analogous to that for P1. In a reachability game, the sure winning region of a player is equal to his/her almost-sure winning region \cite{deAlfaro2007concurrent}. We denote the winning regions of P1 and P2 by $\win_1$ and $\win_2$, respectively.

When P1 is playing a reachability game, an adversarial P2 is playing a \emph{safety game}. In the safety game, P2's objective is to prevent P1 from reaching any state in $F$. P2's winning strategy strategy in such a safety game is called a \emph{permissive strategy} \cite{bernet2002permissive}. Formally, P2's permissive strategy in $\game$ is defined as a function $\zeta: \win_2 \rightarrow \dist{A_2}$ such that for any $a \in \supp(\zeta(s))$, we have $T(s, a) \in \win_2$.


    

    


%
%
\begin{algorithm}[t]
\caption{Zielonka's Recursive Algorithm}
\label{alg:zielonka}
\begin{algorithmic}[1]
    \Function{Zielonka}{$\game, F$}
        \State $Z_0 \gets F$
        \Repeat
            \State $Y_1 \gets \mathsf{Pre}_1(Z_k)$
            \State $Y_2 \gets \mathsf{Pre}_2(Z_k)$
            \State $Z_{k+1} \gets Z_k \cup Y_1 \cup Y_2$
        \Until {$Z_{k+1} = Z_k$}
        \State \Return $\mathsf{Win}_1 = Z_k, \win_2 = S \setminus \win_1$
    \EndFunction
\end{algorithmic}
\end{algorithm}
Algorithm~\ref{alg:zielonka} is the classical algorithm used to compute the sure (almost-sure) winning region for P1 \cite{mcnaughton1993infinite,zielonka1998infinite}. The algorithm uses two sub-procedures defined as follows: 
\begin{subequations}
    \begin{align}
      \mathsf{Pre}_1(U) & = \{v \in V_1 \mid \exists a \in A_1 : T(v, a) \in U \} \label{eq:pre-exists}\\
      \mathsf{Pre}_2(U) & = \{v \in V_2 \mid \forall a \in A_2 : T(v, a) \in U\} \label{eq:pre-forall}
    \end{align}
\end{subequations}
Intuitively, $\mathsf{Pre}_1(U)$ is the set of P1 states at which P1 has an action to lead the game into the given subset of states $U \subseteq S$. Whereas, $\mathsf{Pre}_2(U)$ is the set of P2 states at which every action of P2 leads the game into $U$. Altogether, $\mathsf{Pre}_1(U) \cup \mathsf{Pre}_2(U)$ represents the set of states from where the game enters $U$ within one-step. 
We introduce a running example to explain the concepts discussed this paper.

\begin{example}[Part A]
    \label{ex:1}
    
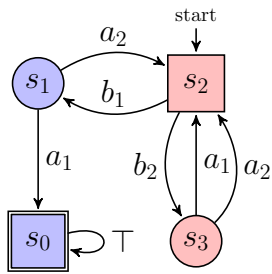
\begin{figure}[htb]
    \centering
    \begin{tikzpicture}[->,>=stealth',shorten >=1pt,auto,node distance=3cm,scale=.7, semithick, transform shape,square/.style={regular polygon,regular polygon sides=4}]
    		\node[accepting, state, square,fill=blue!25!white]     (0)                           {\LARGE$s_0$};
    		\node[state,fill=blue!25!white]                        (1) [above of =0]             {\LARGE$s_1$};
    		\node[initial above, state, square,fill=red!25!white] (2) [right of =1]             {\LARGE$s_2$};
    		\node[state,fill=red!25!white]                        (3) [below of =2]             {\LARGE$s_3$};
    		
    		\path[->]   
    		(0) edge[loop right]    node{\LARGE$\top$}            (0)
    		(1) edge                node{\LARGE$a_1$}             (0)
            (1) edge[bend left]     node{\LARGE$a_2$}             (2)
            (2) edge[bend left]     node[above]{\LARGE$b_1$}      (1)
            (2) edge[bend right]    node[left]{\LARGE$b_2$}       (3)
            (3) edge                node[right]{\LARGE$a_1$}      (2)
            (3) edge[bend right, out=-45, in=-145]    node[right]{\LARGE$a_2$}      (2)
           ;
    	\end{tikzpicture}
    \caption{An example game on graph. The state space is divided into two parts: blue states $\win_1 = \{s_0, s_1\}$ are sure (almost-sure) winning for P1, and red states $\win_2 = \{s_2, s_3\}$ are sure (almost-sure) winning for P2.}
    \label{fig:game-1}
\end{figure}
    Consider the game graph as shown in \fig{fig:game-1}. The circle states $\{s_1, s_3\}$ are P1 states and the square states $\{s_0, s_2\}$ are P2 states. The objective of P1 is to reach to the final states set $F = \{s_0$\} from the initial state $s_2$. P1's action set is $A_1 = \{a_1, a_2\}$ and P2's action set is $A_2 = \{b_1, b_2\}$.

    The sure (or almost-sure) winning region of P1 in the game is $\win_1 = \{s_0, s_1\}$, shown in Fig.~\ref{fig:game-1} as blue states. This is intuitively understood as follows. P1 can win from state $s_1$ by choosing the action $a_1$. However, the states $\win_2 = \{s_2, s_3\}$, shown in Fig.~\ref{fig:game-1} as red states, are losing for P1 because P2 has a strategy to indefinitely restrict the game within $\win_2$ by always selecting action $b_2$ at state $s_2$. 
\end{example}

\section{Games with One-sided Incomplete Information of Action Sets}
    \label{sec:problem-formulation}

In this paper, we study the class of games in which P1 and P2 play with different information about each other's action sets. In particular, we consider the \emph{games with one-sided incomplete information of action sets} with the following information structure.



\begin{assumption}(Information Structure)
    \label{assume:information-asymmetry} 
    Both players have complete information about the game state space $S$, the final states $F$ at all times, and
    \begin{itemize}
        \item P1 has complete information about the action sets of both the players, \ie P1 knows $A_1$ and $A_2$;
        
        \item P2 only knows his own action set $A_2$, but misperceives P1's action set to be a subset $X_0 \subsetneq A_1$ at the beginning of the game;
        
        \item P1 knows $X_0$. 
    \end{itemize}
\end{assumption}

As a consequence of Assumption~\ref{assume:information-asymmetry}, the two players perceive their interaction differently. Given complete information, P1 knows the \emph{true} game, $\langle S, A_1 \cup A_2, T, F \rangle$. Whereas, at the beginning, the game in P2's mind is a misperceived game $\langle S, X_0 \cup A_2, T, F \rangle$ for some $X_0 \subsetneq A_1$. 

\begin{notation}
A game in which P1's perceived action set is $X \subseteq A_1$ is denoted by $\game(X) = \langle S, X \cup A_2, T, F \rangle$. The sure and almost-sure winning regions of P1 and P2 in the game $\game(X)$ are denoted by $\win_1(X)$ and $\win_2(X)$, respectively. We refer to the game $\game(A_1)$, which corresponds to a game with complete, symmetric information, as the \emph{true} game. The winning regions $\win_1(A_1)$ of P1 and $\win_2(A_1)$ of P2 in the \emph{true} game $\game(A_1)$ are called the \emph{non-deceptive} winning regions.
\end{notation}

Notice that Assumption~\ref{assume:information-asymmetry} allows P2's perceptual game to evolve during the interaction. Assuming complete observability, we expect that whenever P1 uses a private action $a \notin X$, P2 would update his perception $X$ to at least include $a$. That is, his updated perception would be a superset of $X \cup \{a\}$. However, it is possible for P2 to add more actions than just $a$ to his current perception. For instance, suppose that P2's perception of P1 is that she can jump 1 or 2 stairs at a time. If P1 jumps 5 stairs during her turn, then P2 can infer that she can also jump 3 and 4 stairs at a time. We formalize such inference capabilities by equipping P2 with an \emph{inference mechanism} defined as follows:

\begin{definition}[Inference Mechanism]       
    \label{def:inference-mechanism} 
    
    A deterministic inference mechanism is a function $\eta: 2^{A_1} \times A_1 \rightarrow 2^{A_1}$ that maps a subset of actions $X \subseteq A_1$ and an action $a \in A_1$ to another subset of actions $Y = \eta(X, a)$ such that $a \in Y$. 
\end{definition}

\setcounter{example}{0}
\begin{example}[Part B]
    \label{ex:2}
    
    Suppose that in \ex{ex:1} (Part A), the action $a_1$ of P1 is a private action. Thus, at the beginning of the interaction, P2's perception of P1's action set is $X_0 = \{a_2\}$ and his perceptual game is the game $\game_2 = \game(X_0)$ as shown in \fig{fig:game-2}. Notice that \fig{fig:game-2} does not include edges corresponding to action $a_1$. On the other hand, P1's perceptual game is same as the \emph{true} game $\game_1 = \game(A_1)$ shown in \fig{fig:game-1}. Given that the final states set $\{s_0\}$ is not reachable in $\game_2$, P2 misperceives both of his actions, $b_1$ and $b_2$, to be safe to play at state $s_2$. However, in reality, only the action $b_2$ is safe in the \emph{true} game, $\game_1$.
    
    \begin{figure}
        \centering
        \begin{tikzpicture}[->,>=stealth',shorten >=1pt,auto,node distance=3cm,scale=.7, semithick, transform shape,square/.style={regular polygon,regular polygon sides=4}]
    		\node[accepting, state, square,fill=blue!20!white]     (0)                           {\LARGE$s_0$};
    		\node[state,fill=red!25!white]                        (1) [above of =0]             {\LARGE$s_1$};
    		\node[initial above, state, square,fill=red!20!white] (2) [right of =1]             {\LARGE$s_2$};
    		\node[state,fill=red!20!white]                        (3) [below of =2]             {\LARGE$s_3$};
    		
    		\path[->]   
    		(0) edge[loop right]    node{\LARGE$\top$}            (0)
            (1) edge[bend left]     node{\LARGE$a_2$}             (2)
            (2) edge[bend left]     node[above]{\LARGE$b_1$}      (1)
            (2) edge[bend right]    node[left]{\LARGE$b_2$}       (3)
            (3) edge[bend right]    node[right]{\LARGE$a_2$}      (2)
           ;
    	\end{tikzpicture}
        \caption{Perceptual game of P2 when he misperceives P1's action set to be $X_0 = \{a_2\}$. The state space is divided into two parts: the blue state $\{s_0\}$ is perceived by P2 as the only winning state of P1, and the red states $\{s_1, s_2, s_3\}$ are perceived by him to be winning for himself. Due to misperception, this partition is different from the partition in Fig.~\ref{fig:game-1}.} 
        \label{fig:game-2}
    \end{figure}
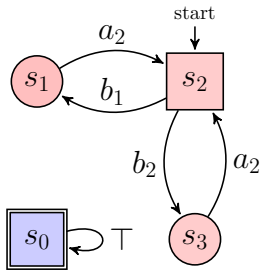
    
    Moreover, when P1 is aware of P2's misperception  $X_0$, she may compute a deceptive strategy which, intuitively, would not use $a_1$ unless the game state is $s_1$. Because, if P1 uses $a_1$ at $s_3$ then P2 will update his perception to $X_1 = A_1$ and conclude that action $b_1$ is unsafe to play at state $s_2$. In this case, P1 will never be able to win the game. 
\end{example}

When P2 is capable of updating her perception during the game, we say P2 has an \textit{evolving perception} of the game. This poses an interesting decision-making problem for P1: \textit{How can P1 improve her strategy in the reachability game if she has the knowledge of P2's initial misperception $X_0$ and his inference mechanism $\eta$?} With this insight, we formally state our problem statement.

\begin{problem} 
    \label{prob:problem-statement} 
    
    Consider a game with one-sided incomplete information in which Assumption~\ref{assume:information-asymmetry} holds and P1 knows P2's inference mechanism, $\eta$. Determine a winning strategy for P1 to satisfy her reachability objective under sure and almost-sure winning conditions.
\end{problem}

When P2's perception is evolving, a winning strategy of P1 must strategize \emph{when to reveal a private action} so as to control P2's perception to her own advantage. We recognize such a behavior to be a deceptive behavior \cite{ettinger2010theory}, and thereby call such a winning strategy to be a \emph{deceptive winning strategy}. Under this notion, we want to investigate whether the use of deceptive winning strategy provides any advantage to P1 over using a non-deceptive winning strategy. Intuitively, the use of deceptive strategy is advantageous for P1 if she has a deceptive winning strategy at some state $s \in S$, at which she \emph{does not} have a non-deceptive winning strategy. 

\section{Dynamic Hypergame on Graph}
    \label{sec:hypergame-model}

In this section, we review the formal definition of a hypergame and then introduce a model, which we call a \textit{dynamic hypergame on graph}, to capture the interaction between P1 and P2 as described in Problem~\ref{prob:problem-statement}.


\subsection{Hypergame Model}
    \label{sec:hypergame-classical}
    
Hypergames are defined inductively based on the level of perception of individual players. A zeroth-level hypergame is a game with complete, symmetric information, where the perceptual games of both players' are identical to the \emph{true} game. In a first-level hypergame, at least one of the players, say P2, misperceives the \emph{true} game but neither of them is aware of it. In this case, both players believe their perceptual game to be the \emph{true} game and play according to their perceptual games, which are zeroth-level hypergames. In a second-level hypergame, one of the players becomes aware of the misperception and is able to reason about her opponent's perceptual game. Recognizing that a second-level hypergame can represent the information structure given in Assumption~\ref{assume:information-asymmetry}, we use a second-level hypergame\footnote{In general, it is possible define hypergames of an arbitrary level. The interested readers may refer to \cite{wang1989solution} for an elaborate discussion on the higher levels of hypergames.} to model Problem~\ref{prob:problem-statement}.


\begin{definition}[Second-level Hypergame for Action Deception]
    \label{def:second-level-hgame}
    
    Let $X \subseteq A_1$ be the action set of P1 as perceived by P2. A second-level hypergame representing the game between P1 and P2 under Assumption~\ref{assume:information-asymmetry} is the tuple,
    \[
        \hgame^2 = \langle \hgame_1^1, \game(X) \rangle,
    \]
    where
    \begin{itemize}
        \item $\hgame_1^1 = \langle \game(A_1), \game(X) \rangle$ is the first-level hypergame being played by P1 in which $\game(A_1)$ is P1's perceptual game and $\game(X)$ is P2's perceptual game;
        \item $\game(X)$ is the zeroth-level hypergame being played by P2.
    \end{itemize}
\end{definition}

While Definition~\ref{def:second-level-hgame} effectively represents the information structure in Problem~\ref{prob:problem-statement} for a fixed perception $X$ of P2, it does not explicitly model the effect of evolving perception on the hypergame $\hgame^2$. To address this limitation, we extend Definition~\ref{def:second-level-hgame} to define a \emph{dynamic hypergame model}. But first, we introduce a graphical model called inference graph to represent the evolution of P2's perceptual game.



\subsection{Inference Graph}
    \label{sec:inference-graph}

\begin{definition}[Inference Graph]
    \label{def:inference-graph}
    
    Given that P2's perceptual game is always an element from the set $\Gamma = \{\game(X_i) \mid X_i \subseteq A_1\}$, an inference graph is defined as a tuple, 
    \[
        \mathcal{I} = \langle \Gamma, E, \gamma_0 \rangle
    \] 
    where
    \begin{itemize}
        \item $\Gamma$ is the set of vertices of $\mathcal{I}$, 
        
        \item $E: \Gamma \times A_1 \rightarrow \Gamma$ defines the set of action-labeled edges of $\mathcal{I}$ such that, $E(\gamma_i, a) = \gamma_j$ for any $\gamma_i = \game(X_i)$,  $\gamma_j = \game(X_j)$ and any $a \in A_1$ if and only if $\eta(X_i, a) = X_j$, where $\eta$ is the inference mechanism of P2, and
        
        \item $\gamma_0 = \game(X_0)$ is the initial perceptual game of P2.
    \end{itemize}
\end{definition}

Intuitively, the nodes of the inference graph represent the possible perceptual games of P2. An edge of the inference graph with a label $a \in A_1$ corresponds to an evolution of P2's perceptual game when he observes P1 using the action $a$. We assume that the inference graph is complete, \ie $E(\gamma, a)$ is defined for any $\gamma = \game(X) \in \Gamma$ and $a \in A_1$. Clearly, $E(\gamma, a) = \gamma$ holds for any P1 action $a \in X$ which is already known to P2.

\subsection{Dynamic Hypergame on Graph}
    \label{sec:dynamic-hypergame-on-graph}
    
Given the notion of an inference graph, we define a dynamic hypergame on graph as a synchronous product of the \textit{true} game and the inference graph.

\begin{definition}[Dynamic Hypergame on Graph]
    \label{def:hgame-on-graph}
    
    Given the \textit{true} game between P1 and P2, $\game(A_1)$, and P2's inference graph, $\mathcal{I}$, the dynamic hypergame on graph is the tuple,
    \[
        \hgame = \game(A_1) \otimes \mathcal{I} = \langle V, \act = A_1 \cup A_2, \Delta, \mathcal{F} \rangle,
    \]
    where 
    \begin{itemize}
        \item $V = S \times \Gamma$  is the set of states in the dynamic hypergame; 
        
        \item $\act = A_1 \cup A_2$ is the set of actions of P1 and P2;
        
        \item $\Delta: V \times \act \rightarrow V$ is the transition function such that, given two states $v = (s, \gamma) \in V$, $v' = (s', \gamma') \in V$ and an action $a \in \act$, we have $\Delta(v, a) = v'$ if and only if $T(s, a) = s'$ and $E(\gamma, a) = \gamma'$; and
        
        \item $\mathcal{F} = F \times \Gamma$ is the set of final states. 
    \end{itemize}
\end{definition}

Hereafter, we refer to a dynamic hypergame on graph as simply a hypergame. Analogous to the game on graph, we define a \textit{hypergame-play} in $\hgame$ as an infinite, ordered sequence of state-action pairs $\tau = v_0 a_0 v_1 a_1 \ldots$ and the action-history as $\alpha = \tau \downharpoonright_\act  = a_0 a_1 \ldots$. In contrast to the reachability game on graphs, we distinguish between (i) a \textit{hypergame-run}, which is the projection of trace onto the hypergame state space $\nu = \tau \downharpoonright_V = v_0 v_1 v_2 \ldots$, and (ii) a \textit{game-run}, which is the projection of trace onto game state space $\rho = \tau \downharpoonright_S = s_0 s_1 s_2 \ldots$, where $s_k$ is the game state corresponding to hypergame state $v_k = (s_k, \gamma)$, for some $\gamma \in \Gamma$. A hypergame-play $\tau$ is said to be winning for P1 when the corresponding hypergame-run $\nu = \tau \downharpoonright_V$ visits the final states $\calF$ in the hypergame, \ie $\occ(\nu) \cap \calF \neq \emptyset$. It follows from the definition of hypergame-play that whenever $\occ(\tau \downharpoonright_V) \cap \calF \neq \emptyset$ then the corresponding game-run $\rho = \tau \downharpoonright_S$ satisfies $\occ(\rho) \cap F \neq \emptyset$.

\setcounter{example}{0}
\begin{example}[Part C]
    \label{ex:idea-example-1} 
    
    The hypergame modeling the asymmetric information from \ex{ex:2} (Part B) is shown in \fig{fig:hgame}. The figure only shows the reachable states. Every state in the hypergame is represented as a tuple of a game state and the current perception of P2 at that state. Given $X_0 = \{a_2\}$, two perceptual games of P2: $\gamma_1 = \game(\{a_2\})$ and $\gamma_2 = \game(\{a_1, a_2\})$, are possible. Any hypergame-play that visits the final state $(s_0, \gamma_2)$ is winning for P1. Therefore, the hypergame-plays $\tau_1 = (s_2, \gamma_1) b_1 (s_1, \gamma_1) a_1 (s_0, \gamma_2)$ and $\tau_2 = (s_2, \gamma_1) b_2 (s_3, \gamma_1) a_1 (s_2, \gamma_2) b_1 (s_1, \gamma_2) a_1 (s_0, \gamma_2)$ are the examples of winning plays for P1. Interestingly, in the next section, we will show that the play $\tau_2$ may never occur if both players act rationally. However, it is possible for the play $\tau_1$ to be observed. 
    
    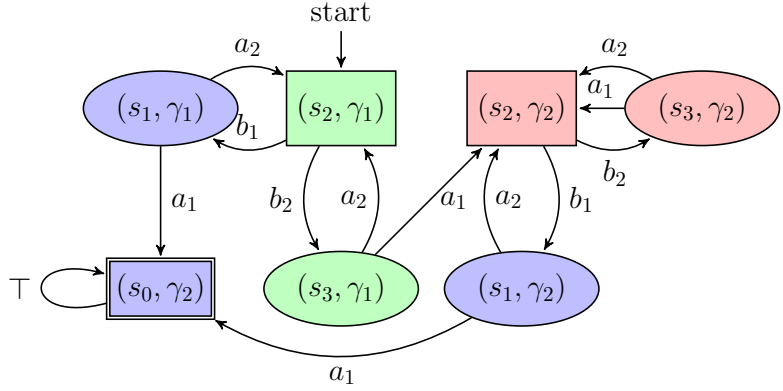
\begin{figure}
        \centering
        \begin{tikzpicture}[->,>=stealth',shorten >=1pt,auto,node distance=3cm,scale=.8,semithick, transform shape,square/.style={rectangle}]
    		\tikzstyle{every state}=[fill=black!10!white]
    		\node[state, accepting, square, draw, fill=blue!25!white]      (0)                           {\Large$(s_0, \gamma_2)$};
    \Large		\node[state, ellipse, draw, fill=blue!25!white]                (1) [above of =0]             {\Large$(s_1, \gamma_1)$};
    		\node[state, initial above, square, draw, fill=green!25!white]  (2) [right of =1]             {\Large$(s_2, \gamma_1)$};
    		\node[state, ellipse, draw, fill=green!25!white]                (3) [below of =2]             {\Large$(s_3, \gamma_1)$};
    		\node[state, ellipse, draw, fill=blue!25!white]                (4) [right of =3]             {\Large$(s_1, \gamma_2)$};
    		\node[state, square, draw,fill=red!25!white]                 (5) [right of =2]             {\Large$(s_2, \gamma_2)$};
    		\node[state, ellipse, draw,fill=red!25!white]                (6) [right of =5]             {\Large$(s_3, \gamma_2)$};

    		
    		\path[->]   
    		(0) edge[loop left]    node{\Large$\top$}                      (0)
    		(1) edge                node{\Large$a_1$}                      (0)
            (1) edge[bend left]     node{\Large$a_2$}                      (2)
            (2) edge[bend left]     node[above]{\Large$b_1$}               (1)
            (2) edge[bend right]    node[left]{\Large$b_2$}                (3)
            (3) edge                node[right]{\Large$a_1$}               (5)
            (3) edge[bend right]    node[left]{\Large$a_2$}                (2)
            (4) edge[bend left]     node[below]{\Large$a_1$}               (0)
            (4) edge[bend left]     node[right]{\Large$a_2$}               (5)
            (5) edge[bend left]     node[right]{\Large$b_1$}               (4)
            (5) edge[bend right]    node[below]{\Large$b_2$}               (6)
            (6) edge                node[above]{\Large$a_1$}               (5)
            (6) edge[bend right]    node[above]{\Large$a_2$}               (5)
           ;
    	\end{tikzpicture}
        \caption{The dynamic hypergame on graph. The state space is divided into three parts: blue states $\{(s_0, \gamma_2), (s_1, \gamma_1), (s_1, \gamma_2)\}$ are sure (almost-sure) winning for P1, and red states $\{(s_2, \gamma_2), (s_3, \gamma_2)\}$ are sure (almost-sure) winning for P2 regardless of whether P1 uses deception or not. The green states $\{(s_2, \gamma_1), (s_3, \gamma_1)\}$ are almost-sure winning, but not sure winning, for P1 when she uses deception.}
        \label{fig:hgame}
    \end{figure}
\end{example}

\section{Synthesis of Deceptive Sure-Winning Strategy}
    \label{sec:dsw-synthesis}

In this section, we address the problem of synthesizing \emph{deceptive sure winning} strategy for P1, given the knowledge of P2's initial misperception, $X_0$, and his inference mechanism, $\eta$. 

\subsection{P2's Rational Strategy}
    \label{subsec:pp-strategy}
    
We start by understanding how a \emph{rational} P2 selects his strategy given his evolving perception. Recall from Section~\ref{sec:preliminaries} that whenever P1 plays a reachability game, P2 plays a safety game in which his winning strategy is given as a permissive strategy \cite{bernet2002permissive}. Intuitively, by following the permissive strategy, P2 is ensured to remain within his winning region. 
However, when P2's perception evolves during the game, his perceived winning region also changes, which means his perceived permissive strategy must also change. To capture this dependence of permissive strategy on the perception of P2, we define the notion of \emph{perceptually permissive strategy}.

\begin{definition}[Perceptually Permissive Action]
    \label{def:perceptually-perm-action}
    
    Given P2's perception $X \subseteq A_1$, an action $a \in A_2$ is said to be a perceptually permissive action for P2 at a state $s \in \win_2(X)$ if and only if the state $s' = T(s, a)$ is winning for P2 under his perception $X$; \ie $s' \in \win_2(X)$. The set of all perceptually permissive actions at the state $s$ is denoted by $M_X(s)$. 
\end{definition}

\begin{definition}[Perceptually Permissive Strategy]
    \label{def:perceptually-perm-strategy}
    
    A \emph{perceptually permissive strategy} of P2 at a state $s \in \win_2(X)$ is a memoryless randomized strategy $\mu_X: \win_2 \rightarrow \dist{A_2}$ such that only perceptually permissive actions have a positive probability to be selected, \ie $\emptyset \subsetneq \supp(\mu_X(s)) \subseteq M_X(s)$.
\end{definition}

Note that a perceptually permissive strategy is a set of randomized strategies that are defined over P2's perceptual game, $\game(X)$. Given that P1 reasons about their interaction a hypergame model, we lift Definitions~\ref{def:perceptually-perm-action} and \ref{def:perceptually-perm-strategy} from the game model to hypergame.

\begin{notation}
    Given a hypergame state $v = (s, \gamma)$ with $s \in \win_2(X)$ and $\gamma = \game(X)$ for some $X \subseteq A_1$, the set of perceptually permissive actions at $v$ is defined as $M(v) = M_X(s)$. The perceptually permissive strategy at $v$ is a distribution $\mu(v) \in \dist{M(v)}$.
\end{notation}

We now establish a result to capture the effect of evolving perception on perceived winning regions of P1 and P2.

\begin{proposition}[Monotonicity Property]
    \label{prop:P1-win-containment}
    
    Given two subsets $X, Y \subseteq A_1$, if $X \subseteq Y$ then $\win_1(X) \subseteq \win_1(Y)$, or equivalently $\win_2(X) \supseteq \win_2(Y)$.
\end{proposition}

\begin{proof}
    Recall that $\win_1(X)$ and $\win_1(Y)$ can be computed using Algorithm~\ref{alg:zielonka} over games $\game(X)$ and $\game(Y)$, respectively. Let $\mathsf{rank}(s)$ be the smallest $k \in \mathbb{N}$ such that $s \in Z_k$ in Algorithm~\ref{alg:zielonka}. Given some $k \geq 0$, let $\win_1^k(X) = \{s \in \win_1(X) \mid \mathsf{rank}(s) \leq k\}$ be the set of P1's sure winning states in the perceptual game $\game(X)$ with rank less than or equal to $k$. Let $K$ denote the maximum rank of any state in $\win_1(X)$. We will show by induction that $\win_1^k(X) \subseteq \win_1^k(Y)$, for all $0 \leq k \leq K$. 
    
    \textbf{Basis:} The statement holds for $k = 0$ because $\win_1^0(X) = \win_1^0(Y) = F$ is true by definition of Algorithm~\ref{alg:zielonka}.
    
    \textbf{Inductive step:} Suppose $\win_1^m(X) \subseteq \win_1^m(Y)$ holds for any $0 \leq m < K$. Then, we must show that $\win_1^{m+1}(X) \subseteq \win_1^{m+1}(Y)$ also holds. To this end, we show that for any $s \in \win_1^{m+1}(X)$, we have $s \in \win_1^{m+1}(Y)$. 
    
    Consider a P1 state, $s \in S_1 \cap \win_1^{m+1}(X)$. By Equation~(\ref{eq:pre-exists}), there exists an action $a \in A_1$ such that $T(s, a) \in \win_1^m(X)$. Let $s' = T(s, a)$. By induction hypothesis, $s' \in \win_1^m(Y)$. Furthermore, if $a \in X$ then $a \in Y$ because $X \subseteq Y$. Hence, there exists an action $a \in Y$ at the state $s$ such that $T(s, a) \in \win_1^m(Y)$. Therefore, $s \in \win_1^{m+1}(Y)$.
    
    Consider a P2 state, $s \in S_2 \cap \win_1^{m+1}(X)$. By Equation~(\ref{eq:pre-forall}), we know that for any action $a \in A_2$, the resulting state $T(s, a)$ is an element of $\win_1^m(X)$. Although P2 misperceives P1's action set, his action set $A_2$ is the same in any game  $\game(X)$. Hence, for all actions $a \in A_2$, by induction hypothesis, we have $T(s, a) \in \win_1^{m}(X) \subseteq \win_1^{m}(Y)$. Therefore, we have $s \in \win_1^{m+1}(Y)$. 
    
    It follows that $\win_1(X) = \win_1^{K}(X) \subseteq \win_1^{K}(Y) \subseteq \win_1(Y)$. Equivalently, we have $\win_2(X) \supseteq \win_2(Y)$. 
    %
        %
\end{proof}


    

\begin{propcorollary}
    \label{cor:P2-SW-action-is-PP}
    
    For any $X \subseteq A_1$, we have $M_{A_1}(s) \subseteq M_{X}(s)$.
\end{propcorollary}

Corollary~\ref{cor:P2-SW-action-is-PP} provides an important insight into P2's evolving misperception. 
It states that (i) P2 never perceives a permissive action in $\game(A_1)$ to be non-permissive in any of the perceptual games in $\Gamma$, and (ii) P2 might perceive some of his non-permissive actions in $\game(A_1)$ to be permissive in $\game(X)$ when $X \subsetneq A_1$. This observation is an important property of action deception which will be useful to prove Theorem~\ref{thm:sw-equals-dsw} in Section~\ref{subsec:dsw-properties}.


\subsection{Deceptive sure winning Strategy}    
    \label{subsec:dsw-properties}
    
Given the notion of a perceptually permissive strategy of P2, we formally define a deceptive sure winning strategy of P1.

\begin{definition}[Deceptive Sure Winning Strategy]
    \label{def:dsw-strategy}
    
    A memoryless strategy $\pi_1 \in \Pi_1$ is said to be a \emph{deceptively sure winning} for P1 at a state $v \in V$ if and only if, for any perceptually permissive strategy $\mu$ of P2 and for every run $\nu \in \mathsf{Outcomes}(v, \pi_1, \mu)$, we have $\occ(\nu) \cap \calF \neq \emptyset$.
\end{definition}

In Definition~\ref{def:dsw-strategy}, P1 reasons only about all possible perceptually permissive strategies of P2, which is in contrast to Definition~\ref{def:sw-strategy} where P1 reasons about all possible strategies of P2. A hypergame state $v \in V$ from which P1 has a deceptively sure winning strategy is called as a \emph{deceptively sure winning state}. The exhaustive set of deceptively sure winning states is called the \emph{deceptively sure winning region}, denoted by $\dswin_1$. Note that deceptive sure winning region is not defined for P2 because he does not know the hypergame, $\hgame$.

\begin{algorithm}[t]
\caption{Deceptive Sure-Wining Region}
\label{alg:dsw}

\begin{algorithmic}[1]
\Function{DSW}{${\mathcal H}$}
    \State $Z_0 = \win_1(A_1) \times \Gamma$
    \Repeat
        \State $Y_1 = \mathsf{DPre}_1(Z_k)$
        \State $Y_2 = \mathsf{DPre}_2(Z_k)$
        \State $Z_{k+1} = Z_k \cup Y_1 \cup Y_2$
    \Until {$Z_{k+1} = Z_k$}
    \State \Return $\dswin_1 = Z_k$
\EndFunction
\end{algorithmic}
\end{algorithm}

Importantly, we note that the strategy defined in Definition~\ref{def:dsw-strategy} is deceptive. This is because P1's deceptively sure winning strategy makes a conscious decision about when to reveal which private action to P2 during their interaction. 


By observing that (i) every perceptually permissive strategy of P2 is an element of $\Pi_2$, and (ii) a non-deceptive sure winning strategy of P1 ensures the completion of her reachability objective against any strategy in $\Pi_2$, we derive the following result.

\begin{proposition}
    \label{prop:sw-state-is-also-dsw}
    
    If a game state $s \in S$ is a non-deceptive sure winning state for P1 then, for any $\gamma \in \Gamma$, the hypergame state $v = (s, \gamma)$ is a deceptively sure winning state for P1.
\end{proposition}


Algorithm~\ref{alg:dsw} computes the deceptively sure winning region of P1 given the hypergame $\hgame$. It is derived from Algorithm~\ref{alg:zielonka} by adopting the definitions of $\mathsf{Pre}_1(U)$ and $\mathsf{Pre}_2(U)$ as follows:
\begin{subequations}
    \begin{align}
      \mathsf{DPre}_1(U) & = \{v \in V_1 \mid \exists a \in A_1 : \Delta(v, a) \in U \} \label{eq:dpre-exists}\\
      \mathsf{DPre}_2(U) & = \{v \in V_2 \mid \forall a \in M(v) : \Delta(v, a) \in U\} \label{eq:dpre-forall}
    \end{align}
\end{subequations}
Intuitively, $\mathsf{DPre}_1(U)$ is the set of P1 states at which P1 has an action to lead the game into the given subset of states $U \subseteq V$. Whereas, $\mathsf{DPre}_2(U)$ is the set of P2 states at which every perceptually permissive action of P2 leads the game into $U$. Altogether, $\mathsf{DPre}_1(U) \cup \mathsf{DPre}_2(U)$ represents the set of states from where the game enters $U$ within one-step.

Algorithm~\ref{alg:dsw} is initialized with the set $Z_0 = \win_1(A_1) \times \Gamma$ because P1 has a non-deceptive sure winning strategy to complete her reachability objective from any state in $Z_0$. In the $i$-th iteration, the algorithm identifies the states to be added to $Z_{i+1}$ such that, from each of the newly added states, the game is ensured to enter $Z_i$ in one-step. The loop terminates when a fixed-point is reached; \ie when no new states can be added to $Z_{i+1}$.


Next, we show that, for any $X \subseteq A_1$, the set of states in $\game(X)$ from which P1 has a deceptive sure winning strategy is identical to the set of states in $\game(A_1)$ from which she has a non-deceptive sure winning strategy.

\begin{theorem}
    \label{thm:sw-equals-dsw}
    
    Let $\dswin_1 \downharpoonright_S = \{s \in S \mid v \in \dswin_1 \text{ and } s = v \downharpoonright_S \}$ be the set of projection of the deceptively sure winning states onto the game state space. It holds that $\win_1(A_1) = \dswin_1 \downharpoonright_S$. 

\end{theorem}

\begin{proof}
Given $Z_0 = \win_1(A_1) \times \Gamma$, to establish that $\dswin_1 \downharpoonright_S = \win_1(A_1)$ we will show that $\mathsf{DPre}_1(Z_0)$ and $\mathsf{DPre}_2(Z_0)$ are empty at the end of first iteration of Algorithm~\ref{alg:dsw}. 


\textbf{Case I} $\mathbf{(\mathsf{DPre}_1(Z_0) = \emptyset)}$. By contradiction. Suppose there exists a hypergame state $v = (s,i) \in V_1 \setminus Z_0$ that is added to $\dswin_1$ in the first iteration. Then, by Eq.~(\ref{eq:dpre-exists}), there exists an action $a \in A_1$ such that $\Delta(v, a) \in Z_0$. But this would mean $T(s, a) \in \win_1(A_1)$ which in turn implies that the state $s$ is a sure winning state of P1. Thus, the hypergame state $v$ must be in $Z_0$---a contradiction.

\textbf{Case II} $\mathbf{(\mathsf{DPre}_2(Z_0) = \emptyset)}$. We will show that at every state $v = (s, \gamma) \in V \setminus Z_0$, P2 has a perceptually permissive action $a^\ast \in A_2$ such that $\Delta(v, a^\ast) \notin Z_0$. To see this, first, we note that $v \in V \setminus Z_0$ implies that $s \notin \win_1(A_1)$ from case I. Second, we recall that whenever $s \notin \win_1(A_1)$, we have $s \in \win_2(A_1)$. This implies that there exists an action $a^\ast \in A_2$ such that $T(s, a^\ast) \in \win_2(A_1)$. As $\win_2(X) \supseteq \win_2(A_1)$ holds for any subset $X$ of $A_1$, the action $a^\ast$ must be a permissive action at $v = (s, \gamma)$ as long as $\gamma \ne \game(A_1)$. Thus, we conclude by Eq.~(\ref{eq:dpre-forall}) that $\mathsf{DPre}_2(Z_0) = \emptyset$.  
\end{proof}

Theorem~\ref{thm:sw-equals-dsw} states that P1 gains \emph{no} advantage by using action deception under the sure winning condition. Given that P1's sure and almost sure winning regions are equal \cite{deAlfaro2007concurrent}, we note that P1's non-deceptive sure and almost-sure winning regions are equal to her deceptive sure winning region. In other words, P1's non-deceptive sure, almost-sure strategies and deceptive sure winning strategy are all equally powerful. We revisit our running example to illustrate our conclusion. 


\setcounter{example}{0}
\begin{example}[Part D]
    \label{ex:4}
    
    Consider the hypergame shown in \fig{fig:hgame}. Recall from \ex{ex:1} (Part A) that sure winning region of P1 is $\win_1(A_1) = \{s_0, s_1\}$. Therefore, following Proposition~\ref{prop:sw-state-is-also-dsw}, we have $Z_0 = \{(s_0, \gamma_2), (s_1, \gamma_2), (s_1, \gamma_1)\}$ (we omit $(s_0, \gamma_1)$ as it is unreachable). Consider the states $(s_2, \gamma_1)$ and $(s_2, \gamma_2)$. At $(s_2, \gamma_2)$, P2's perceptual game is $\game(\{a_1, a_2\})$. Thus, his perceptually permissive winning strategy at $(s_2, \gamma_2)$ is $\mu((s_2, \gamma_2)) = b_2$. At $(s_2, \gamma_1)$, P2's perceptual game is $\game(\{a_2\})$. Therefore, he has two perceptually permissive actions at $(s_2, \gamma_1)$: $b_1$ and $b_2$. Given that the action $b_2$ is losing for P1 by Definition~\ref{def:dsw-strategy}, both the states $(s_2, \gamma_1)$ and $(s_3, \gamma_1)$ are not deceptively sure winning for P1.

\end{example}

In above example, we see that when P1 uses the sure winning condition, she considers the worst-case strategy of P2. However, from P2's perspective, he is indifferent to using any actions in $\{b_1, b_2\}$ at $(s_2, \gamma_1)$. In other words, P2 may choose either action with some positive probability. In such a case, when P2 uses a randomized strategy, we want to know whether the use of deception is advantageous to P1 or not? In the next section, we answer this question positively. 

\section{Synthesis of Deceptive Almost Sure-Winning Strategy}
    \label{sec:dasw-synthesis}

We start by defining a deceptive almost-sure winning strategy in an analogous way to deceptive sure winning strategy, \ie by adapting Definition~\ref{def:asw-strategy} to the hypergame. 

\begin{assumption}
    \label{assume:PPStrategy-Support}
    
    P2 plays a randomized perceptually permissive strategy $\mu$ such that for all $v \in \win_2$, we have $\supp\left(\mu\left( v \right)\right) = M(v)$.
\end{assumption}

Assumption~\ref{assume:PPStrategy-Support} states that every perceptually permissive action at a given state can be chosen by P2 with a positive probability. Given this assumption, we will identify the set of states from which P1 has a deceptive almost-sure winning strategy that leverages Corollary~\ref{cor:P2-SW-action-is-PP} to almost-surely satisfy her reachability objective.


\begin{definition}
    \label{defn:dasw} 
    
    A memoryless strategy $\pi_1 \in \Pi_1$ is said to be \textit{deceptively almost-sure winning} for P1 at a state $v \in V$ if and only if, for any perceptually permissive strategy $\mu$ of P2 satisfying Assumption~\ref{assume:PPStrategy-Support} and for every run $\nu \in \mathsf{Outcomes}(v, \pi_1, \pi_2)$, we have that $Pr_v^{\pi_1, \mu}(\calF) = 1$. 
\end{definition}

A state $v \in V$ from which P1 has a deceptive almost-sure winning strategy is called as a \emph{deceptive almost-sure winning state}. The exhaustive set of deceptive almost-sure winning states is called the \emph{deceptive almost-sure winning region}, and is denoted by $\dawin_1$. 

We propose Algorithm~\ref{alg:DASW} to compute the deceptive almost-sure winning region for P1. Our algorithm is inspired by the algorithm presented in \cite{deAlfaro2007concurrent} to compute the almost-sure winning region in a concurrent $\omega$-regular games. The idea behind Algorithm~\ref{alg:DASW} is to identify and exploit the states $v = (s, \gamma)$ at which P2's perceptually permissive actions $M(v)$ includes some of his non-permissive actions in the \emph{true} game, $\game(A_1)$. To this end, we define the following sub-routines: 
\begin{subequations}
    \begin{align}
        \dapre_1^1(U) &= \{v \in V_1 \mid \exists a \in A_1 \text{ s.t. } \Delta(v, a) \in U\}, \label{eq:dapre-11} \\ 
        \dapre_1^2(U) &= \{v \in V_2 \mid \forall b \in M(v) \text{ s.t. } \Delta(v, b) \in U\} , \label{eq:dapre-12}\\
        \dapre_2^1(U) &= \{v \in V_1 \mid \forall a \in A_1 \text{ s.t. } \Delta(v, a) \in U\}, \label{eq:dapre-21} \\
        \dapre_2^2(U) &= \{v \in V_2 \mid \forall b \in M(v) \text{ s.t. } \Delta(v, b) \in U\}. \label{eq:dapre-22}
    \end{align}
\end{subequations}
%
%
%
%
%

\begin{proposition}
    \label{prop:asw-state-is-also-dasw}
    
    If a game state $s \in S$ is a non-deceptive almost-sure winning state for P1 then, for any $\gamma \in \Gamma$, the hypergame state $v = (s, \gamma)$ is a deceptively almost-sure winning state for P1.
\end{proposition}

\begin{algorithm}
    \caption{Computation of the \ac{dasw} region and strategy for P1}
    \label{alg:DASW}
    \begin{algorithmic}[1]
        \Function{DASW}{$\hgame$}
            \State $Z_0 = \win_1(A_1) \times \Gamma$
            \While{True}
                \State $C_k = \textsc{Safe-2}(V \setminus Z_k)$ 
                \State $Z_{k+1} = \textsc{Safe-1}(V \setminus C_k)$ 
                \If{$Z_{k+1} = Z_{k}$}
                    \State End loop
                \EndIf
            \EndWhile
            \State \Return $Z_k$
        \EndFunction
    \end{algorithmic}
    \vspace{0.5em}
    \begin{algorithmic}[1]
        \Function{\textsc{Safe-$i$}}{$U$}
            \State $Y_0 = U$
            \While{True} 
                \State $W_1 = \dapre_i^1(Y_j)$
                \State $W_2 = \dapre_i^2(Y_j)$
                \State $Y_{j+1} = Y_j \cap (W_1 \cup  W_2)$
                \If{$Y_{j+1} = Y_{j}$}
                    \State End loop
                \EndIf
            \EndWhile
            \State \Return{$Y_j$}
        \EndFunction
    \end{algorithmic}
\end{algorithm}

Algorithm~\ref{alg:DASW} works as follows. Following Proposition~\ref{prop:asw-state-is-also-dasw}, we initialize the algorithm with $Z_0 = \win_1(A_1) \times \Gamma$ and then iteratively compute the sets $C_k$ and $Z_{k+1}$ for $k = 0, 1, \ldots$ until a fixed-point is reached. 
In the $k$-th iteration, the set $C_k \subseteq V \setminus Z_k$ is computed using sub-routine \textsc{Safe-2}, which identifies the subset of states in $V \setminus Z_k$ from which P1 has no strategy to exit $V \setminus Z_k$. In other words, $C_k$ is a set of states in which P2 can enforce P1 to stay.  
The sub-routine \textsc{Safe-2} starts with $Y_0 = V \setminus Z_k$ and iteratively computes $Y_j$ for $j = 0, 1, \ldots$ by identifying (i) $W_1 = \dapre_2^1(Y_j)$: P1 states within $Y_j$, from which any action $a \in A_1$ leads to a state in $Y_j$, and (ii) $W_2 = \dapre_2^2(Y_j)$: P2 states within $Y_j$, from which any of his perceptually permissive action $a \in M(v)$ leads to a state in $Y_j$. 
Next, the set $Z_{k+1}$ is computed using the sub-routine \textsc{Safe-1}, which identifies the subset of states in $V \setminus C_k$ from which P1 is ensured to visit $Z_k$ in one-step. The sub-routine \textsc{Safe-1} starts with $Y_0 = V \setminus C_k$ and iteratively computes $Y_j$ for $j = 0, 1, \ldots$ by identifying (i) $W_1 = \dapre_1^1(Y_j)$: P1 states within $Y_j$ from which she has an action to enter $Y_j$ in one step, and (ii) $\dapre_1^2(Y_j)$: P2 states within $Y_j$ from which any perceptually permissive action of P2 leads to a state in $Y_j$. It is observed that as $k$ increases, the set $C_k$ shrinks while the set $Z_k$ expands. Intuitively, this is because the states in $C_k$ may have transitions leading outside $C_k$, while remaining within $V \setminus Z_k$. If a state, say $v \in V \setminus Z_k$ that is not in $C_k$, is included in $Z_{k+1}$, then all states in $C_k$ that have a transition going to $v$ are excluded from $C_{k+1}$ and have a potential to be included in $Z_{k+2}$. However, once the fixed-point is reached, say in iteration $K$, we show that all deceptively almost-sure winning states of P1 are included in $Z_K$. A deceptively almost-sure winning strategy can then be computed based on the proof of Theorem~\ref{thm:DASW-sufficiency}.

\setcounter{example}{0}
\begin{example}[Part E]
    \label{ex:5}
     
    In contrast to Example~\ref{ex:4} (Part D), in this part we show that the state $(s_2, \gamma_1)$, which was \emph{not} deceptively sure winning for P1, is a deceptively almost-sure winning state for her. Intuitively, this is because when the game is stuck in a loop between the states $(s_2, \gamma_1)$ and $(s_3, \gamma_1)$, Assumption~\ref{assume:PPStrategy-Support} guarantees that the perceptually permissive action $b_1$ at $(s_2, \gamma_1)$ will eventually be selected. In other words, the game will eventually reach the state $(s_1, \gamma_1)$, from which P1 can win the game by revealing her private action, $a_1$. With this intuition, we describe how the Algorithm~\ref{alg:DASW} identifies $(s_2, \gamma_1)$ as a deceptively almost-sure winning state of P1.

    \subparagraph{Iteration 1 of Algorithm~\ref{alg:DASW}.} The first step is to compute $C_0$, \ie the subset of $V \setminus Z_0$ from which P2 can enforce P1 to remain within $V \setminus Z_0$. The \textsc{Safe-2} sub-routine takes 3 iterations to reach a fixed-point, at the end of which $C_0 = \{(s_2, \gamma_2), (s_3, \gamma_2)\}$. The next step is to compute $Z_1$, which the largest subset of $V \setminus C_0$ in which P1 can stay indefinitely. The \textsc{Safe-1} sub-routine takes 2 iterations to reach a fixed point. In its first iteration, $\dapre_1^1$ adds a state $(s_3, \gamma_1)$ and $\dapre_1^2$ adds a state $(s_2, \gamma_1)$ to $Z_1$. We note that $(s_2, \gamma_1)$ is added because the actions $b_1$ and $b_2$ are perceptually permissive actions for P2, both of which lead to a state in $V \setminus C_0$. 
    
    \subparagraph{Iteration 2 of Algorithm~\ref{alg:DASW}.} The fixed-point of \textsc{DASW} algorithm is reached in this iteration with $Z_2 = \{(s_0, \gamma_2),$ $(s_1, \gamma_1), (s_1, \gamma_2), (s_2, \gamma_1), (s_3, \gamma_1)\}$. 
\end{example}

Given the intuition about the Algorithm~\ref{alg:DASW}, we first note the existence of a deceptively almost-sure winning state that is \emph{not} a non-deceptive almost-sure winning state for P1. Clearly, to win from such a state, P1 \textit{must} use action deception.

\begin{theorem}
    \label{thm:DASW-not-empty}
    
    The deceptive almost-sure winning region may contain a state $v = (s, \gamma)$ such that the state $s \in S$ is \emph{not} a non-deceptive almost-sure winning state, \ie $s \notin \win_1(A_1)$. 
\end{theorem}

\begin{proof}[\textbf{Proof}]
    See Example~\ref{ex:5} (Part E). 
\end{proof}

Next, we establish the correctness of Algorithm~\ref{alg:DASW} by showing that from every state that is identified by the algorithm as a deceptive almost-sure winning state, we can construct a deceptive almost-sure winning strategy for P1 to ensure a visit to a final state with probability \emph{one}.

\begin{lemma}
    \label{lma:stay-in-strategy} 
    
    In the $i$-th iteration of Algorithm~\ref{alg:DASW}, for all states in $Z_i$, P1 has a strategy to restrict the game indefinitely within $Z_i$. 
\end{lemma}

\begin{proof}
    For a P2 state $v \in V_2$, by Eq.~(\ref{eq:dapre-12}), we have that $\Delta(v, a) \in Z_i$ for any perceptually permissive action $a \in \supp(\mu(v))$. For a P1 state $v \in V_1$, by Eq.~(\ref{eq:dapre-11}), there exists an action $a \in A_1$ such that $\Delta(v, a) \in Z_i$. Thus, at any P1 state, P1 has a strategy to enforce a visit to $Z_i$ and, at a P2 state, any perceptually permissive action of P2 leads the game into $Z_i$.
\end{proof}

\begin{lemma}
    \label{lma:progressive-action} 
    
    Every state $v \in V$ that is newly added to $Z_{i+1}$ in the $i$-th iteration of Algorithm~\ref{alg:DASW} has an action leading into $Z_{i}$.
\end{lemma}

\begin{proof}
    From the sub-routine \textsc{Safe-1} in Algorithm~\ref{alg:DASW}, we know that every new state added to $Z_{i+1}$ must be a state in $V \setminus C_i$. But every state in $V \setminus C_i$ has at least one transition leading outside $V \setminus Z_i$. This follows from the fact that the sub-routine \textsc{Safe-2} includes only those P1 states in $C_i$ for which there exists $a \in A_1$ such that $\Delta(v, a) \in V \setminus Z_i$, by Eq.~(\ref{eq:dapre-21}). And it includes only those P2 states in $C_i$ for which $\Delta(v, a) \in V \setminus Z_i$ holds for any $a \in \supp(\mu(v))$, by Eq.~(\ref{eq:dapre-22}). Thus, whenever a state is \emph{not} included in $C_i$ (\ie it belongs to $V \setminus C_i$), there exists an action for P1 or a perceptually permissive action for P2 which leads the game outside $V \setminus Z_i$ (\ie into $Z_i$). 
\end{proof}

The following observation follows immediately from Lemma~\ref{lma:progressive-action}.

\begin{corollary}
    \label{lma:}
    For every $i \geq 0$, we have $Z_i \subseteq Z_{i+1}$.
\end{corollary}

From \lma{lma:progressive-action}, it is easy to see that P1 has a strategy to reach $Z_i$ from a state added to $Z_{i+1}$ in one-step. However, this is not true for P2. From a P2 state in $Z_{i+1}$, there exists a positive probability to reach $Z_i$ because of Assumption~\ref{assume:PPStrategy-Support}. In the next theorem, we prove a stronger statement which states that from every state in $Z_{i+1}$, P1 can not only reach $Z_i$ with positive probability, but with probability one.

\begin{theorem}
    \label{thm:DASW-sufficiency} 
    
    From every deceptively almost-sure winning state $s \in \dawin_1$, P1 has a deceptively almost-sure winning strategy. 
\end{theorem}

\begin{proof}
    The proof follows from \lma{lma:stay-in-strategy} and \lma{lma:progressive-action}. For any $v \in Z_i$, \lma{lma:stay-in-strategy} ensures that P1 has a strategy to stay within $Z_i$ indefinitely. In addition,  \lma{lma:progressive-action} ensures that the probability of reaching to a state $v' \in Z_{i-1}$ from $v$ is strictly positive. Therefore, given a run of infinite length, the probability of reaching $Z_{i-1}$ from $Z_i$ is one. By repeatedly applying the argument, it follows that the probability of reaching $Z_0$ from $Z_i$ is one.
\end{proof}

The deceptively almost-sure winning strategy can be constructed based on the proof of Theorem~\ref{thm:DASW-sufficiency}. Specifically, any randomized strategy $\pi_1 \in \Pi_1$ such that $\supp(\pi_1(v)) = \{a \in A_1 \mid v' = \Delta(v, a) \text{ and } v' \in Z_{i-1}\}$ for any $v \in V_1$, given that $i \geq 1$ is the smallest integer such that $v \in Z_i$, is a deceptive almost-sure winning strategy of P1.


From Theorems~\ref{thm:DASW-not-empty} and \ref{thm:DASW-sufficiency}, we conclude that a deceptively almost-sure winning strategy of P1 is more powerful than the almost-sure winning strategy.

\section{Experiment}
    \label{sec:experiment}

In this section, we illustrate the advantages of using action deception using a simplified version of capture-the-flag game \cite{Svabensky2021} played over a $5 \times 5$ gridworld, like the one shown in Figure~\ref{fig:capture-the-flag}. The gridworld is partitioned into P1 (blue) and P2 (red) territories. P1's objective in the game is to capture both the flags from P2's territory, while that of P2 is to prevent P1 from capturing the flags. We restrict P2 to move only within his own territory. Under this setting, we are interested to determine the number of game states from which P1 has a deceptive sure (almost-sure) winning strategy and compare it with the sizes of the non-deceptive sure (almost-sure) winning regions. We introduce the following notion of \emph{value of deception}, denoted by $\mathsf{VoD}$ to quantify the advantage gained by P1 by using deception. 
%
\begin{align}
    \label{eq:vod}
    \mathsf{VoD} = 
    \begin{cases}
        \frac{|\dswin_1\downharpoonright_S| - |\win_1(A_1)|}{|\win_2(A_1)|} & \text{under deceptive sure winning condition} \\
        \frac{|\dawin_1\downharpoonright_S| - |\win_1(A_1)|}{|\win_2(A_1)|} & \text{under deceptive almost-sure winning condition} \\
        0 & \text{if } |\win_2(A_1)| = 0 \\
    \end{cases}
\end{align}
To understand Eq.~\eqref{eq:vod}, first, recall that P1 can win from any state in $\win_1(A_1)$ regardless of whether she uses deception or not. Thus, the benefit of deception can be quantified by counting the number of P2's winning states in the game with complete, symmetric information (\ie~in $\win_2(A_1)$) that P1 can win from by using deception. Notice that $\mathsf{VoD}$ takes a value between $0$ and $1$. $\mathsf{VoD} = 0$ represents the case when P1 gains no advantage by using deception. $\mathsf{VoD} = 1$ represents the case in which P1 gains maximum benefit that is possible by using deception, \ie P1 can leverage P2's misperception to win from all of P2's winning states in $\win_2(A_1)$. 


To demonstrate the applicability of our proposed approach to a broad range of reachability objectives, we specify P1's objective using a Syntactically Co-safe Linear Temporal Logic (scLTL) formula. scLTL is a subclass of Linear Temporal Logic (LTL) which can represent complex and temporally extended co-safety objectives. An overview of strategy synthesis with scLTL is provided in Appendix~\ref{appendix:ltl}. We consider the following two scLTL objectives for P1 in this experiment.

\begin{enumerate}
    \item P1 must capture both $\mathsf{FLAG}_1$ and $\mathsf{FLAG}_2$ in any order. 
    \begin{align}
    \label{eq:varphi1}
        \underbrace{\Eventually \mathsf{FLAG}_1}_{\text{Eventually capture $\mathsf{FLAG}_1$}} \land \underbrace{\Eventually \mathsf{FLAG}_2}_{\text{Eventually capture $\mathsf{FLAG}_2$}} 
    \end{align}

    \item P1 must first capture $\mathsf{FLAG}_1$ and then capture $\mathsf{FLAG}_2$. Until then, P1 must avoid colliding with P2.
    \begin{align}
    \label{eq:varphi2}
        \underbrace{(\neg\mathsf{FLAG}_2 \land \neg \mathsf{collide}) \until \mathsf{FLAG}_1}_{\text{don't collide or collect $\mathsf{FLAG}_2$ until $\mathsf{FLAG}_1$ is collected}} \land \underbrace{\neg \mathsf{collide} \until \mathsf{FLAG}_2}_{\text{don't collide until $\mathsf{FLAG}_2$ is collected}}
    \end{align}
\end{enumerate}

The dynamics of the capture-the-flag game are as follows. Both the players can move in 4 compass directions: \texttt{N, E, S, W}. P2 cannot enter any cell containing a wall or a fence, and \textit{presumes this to be the case for P1 as well}. However, initially unknown to P2, P1 has the following special actions: \texttt{JumpN, JumpE, JumpS, JumpW} and \texttt{Cut}. Using the \texttt{Jump} action P1 can jump over a wall in a free cell (\ie a cell not containing an obstacle, a wall or a fence) adjacent to the wall in the direction of the jump. Using the \texttt{Cut} action, P1 can convert a cell containing a fence into a free cell. Note that once a cell containing a fence becomes free, P2 can visit that cell.

\begin{figure}[ht]
    \centering
    \includegraphics[scale=0.40]{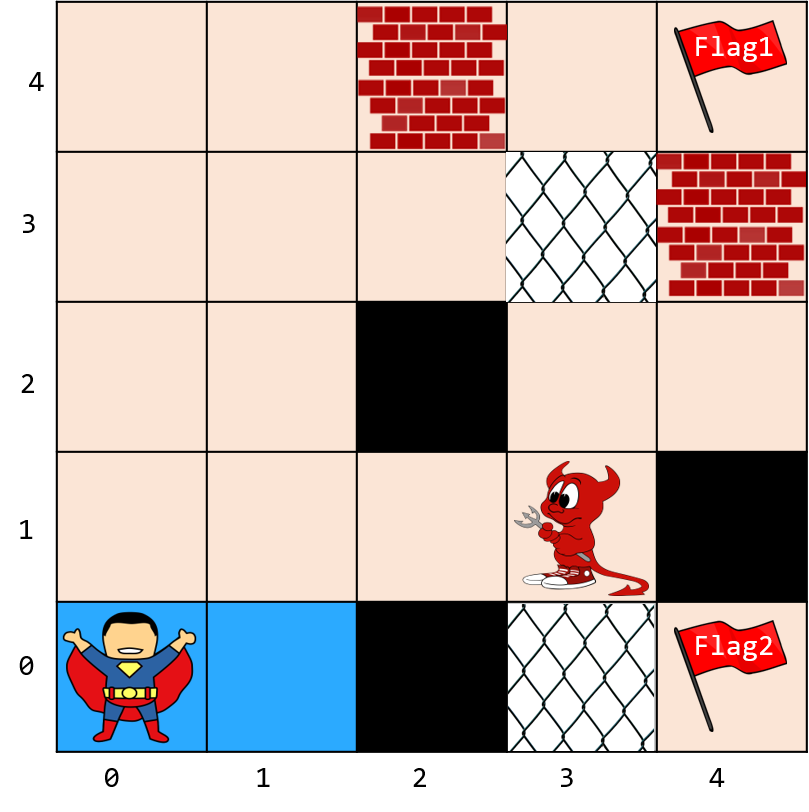}
    \caption{An example of capture-the-flag game between P1 (superman) and P2 (devil) played over a $5 \times 5$ grid world.}
    \label{fig:capture-the-flag}
\end{figure}

Given the dynamics, we construct game and hypergame graphs. We define the game state (denoted by $s$) and hypergame state (denoted by $v$) as follows: 
\begin{align*}
    & s: \big((\mathtt{p1.x, p1.y, p2.x, p2.y}), (\mathtt{f1.cut}, \mathtt{f2.cut}), \mathtt{turn}, \mathtt{q}\big) \\
    & v: \big((\mathtt{p1.x, p1.y, p2.x, p2.y}), (\mathtt{f1.cut}, \mathtt{f2.cut}), \mathtt{turn}, \mathtt{q, i}\big)
\end{align*}
where 
\begin{itemize}
    \item \texttt{p1.x, p2.y, p1.x, p2.y} represents the position of P1 and P2 in gridworld;
    \item \texttt{f1.cut, f2.cut} represents whether fence 1 and fence 2 (cells $(0, 3)$ and $(3, 3)$ in Figure~\ref{fig:capture-the-flag}) are cut or intact;
    \item \texttt{turn} represents whether it is P1's or P2's turn at that state;
    \item \texttt{q} is the specification DFA state that encodes the progress P1 has made towards satisfying her scLTL objective (see Appendix~\ref{appendix:ltl} for more details);
    \item \texttt{i} is a state of inference graph that captures P2's current perception of P1's action set.
\end{itemize}

\begin{table}[b!] 
\renewcommand\arraystretch{1.65}
\begin{tabularx}{\columnwidth}{@{}c|c|c|c|c|c|c|c@{}}
\toprule
&$|V|$&$|E|$&$|F|$&\begin{tabular}[c]{@{}c@{}}$|\mathsf{DSWin}_1|$ or \\ $|\dawin_1|$\end{tabular}&\begin{tabular}[c]{@{}c@{}}$|\mathsf{DSWin}_1\downharpoonright_S|$ or \\ $|\dawin_1\downharpoonright_S|$\end{tabular}&$\win_2$&$\mathsf{VoD}$\\
\midrule
SW($\game$) & 6388 & 15016 & 1686 & - & 6133 & 255 & -\\
DSW($\hgame$) & 9423 & 22181 & 2238 & 9031 & 6133 & 255 & 0 \\
DASW($\hgame$)    & 9423 & 22181 & 2238 & \textbf{9395} & \textbf{6370} & \textbf{18} & \textbf{0.9294}\\
\bottomrule
\end{tabularx}
\caption{Comparison of deceptive and non-deceptive winning states under sure and almost-sure winning condition for P1's objective $\varphi_1 = \Eventually \mathsf{FLAG}_1 \land \Eventually \mathsf{FLAG}_2$.}
\label{tbl:results-1}
\end{table}

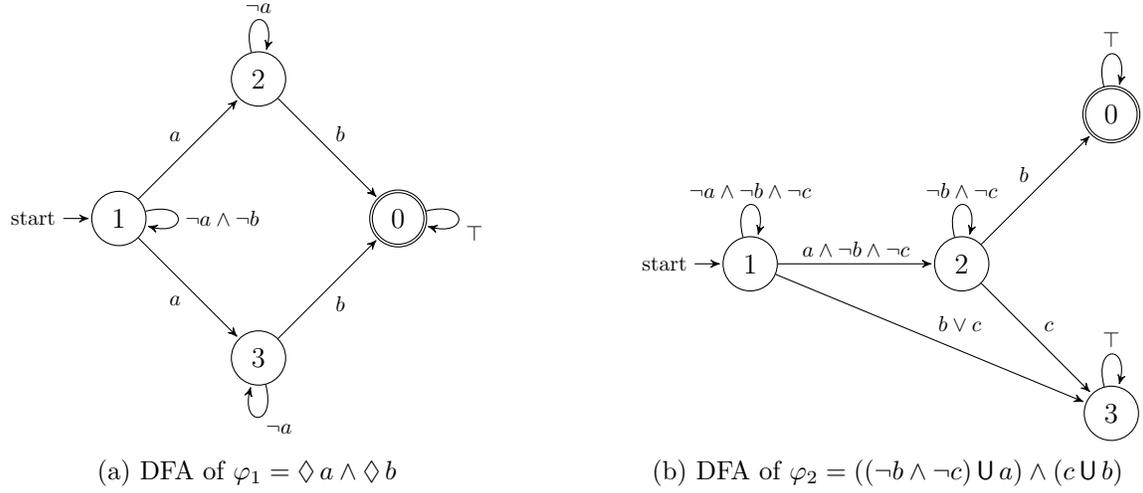
\begin{figure}
     \centering
     \subfloat[DFA of $\varphi_1 = \Eventually a \land \Eventually b$\label{fig:arena-p1}]{%
      \begin{tikzpicture}[->,>=stealth',shorten >=1pt,auto,node distance=3.5cm, scale = 0.75,transform shape]
    
      \node[state, initial left] (1) []             {\Large $1$};
      \node[state] (2) [above right of=1]             {\Large $2$};
      \node[state] (3) [below right of=1]   {\Large $3$};
      \node[state,accepting] (0) [below right of=2]   {\Large $0$};
      
      \path (1) edge   node {$a$} (2)
            (2) edge   node {$b$} (0)
            (1) edge   node[below left] {$a$} (3)
            (3) edge   node[below right] {$b$} (0)
            
            (0) edge[loop right]   node[below right] {$\top$} (0)
            (3) edge[loop below]   node[below right] {$\neg a$} (0)
            (2) edge[loop above]   node[above] {$\neg a$} (2)
            (1) edge[loop right]   node[right] {$\neg a \land \neg b$} (1)
            ;
            
    \end{tikzpicture}
    }\hfill
    \subfloat[DFA of $\varphi_2 = ((\neg b \land \neg c) \until a) \land (c \until b)$\label{fig:arena-p2}]{%
      \begin{tikzpicture}[->,>=stealth',shorten >=1pt,auto,node distance=3.75cm, scale = 0.75,transform shape]
    
      \node[state, initial left] (1) []             {\Large $1$};
      \node[state] (2) [right of=1]             {\Large $2$};
      \node[state] (3) [below right of=2]   {\Large $3$};
      \node[state,accepting] (0) [above right of=2]   {\Large $0$};
      
      \path (1) edge[loop above]   node {$\neg a \land \neg b \land \neg c$} (1)
            (1) edge   node {$a \land \neg b \land \neg c$} (2)
            (1) edge   node {$b \lor c$} (3)
            (2) edge[loop above]   node {$\neg b \land \neg c$} (2)
            (2) edge   node {$b$} (0)
            (2) edge   node {$c$} (3)
            (0) edge[loop above]   node {$\top$} (0)
            (3) edge[loop above]   node {$\top$} (3)
            ;
            
    \end{tikzpicture}
    }
        \caption{The sub-figure (a) shows the DFA equivalent to the scLTL formula given in Equation~\ref{eq:varphi1} and sub-figure (b) shows the DFA equivalent to scLTL formula in Equation~\ref{eq:varphi2}. For brevity, we use $a = \mathsf{FLAG}_1$, $b = \mathsf{FLAG}_2$ and $c = \mathsf{collide}$ in the figure.}
        \label{fig:dfa}
\end{figure}

\begin{figure}
     \centering
     
      \begin{tikzpicture}[->,>=stealth',shorten >=1pt,auto,node distance=4.5cm, scale = 0.7,transform shape]
    
      \node[state, initial left] (1) []             {\Large $0$};
      \node[state] (2) [above right of=1]             {\Large $1$};
      \node[state] (3) [below right of=1]   {\Large $2$};
      \node[state,accepting] (0) [below right of=2]   {\Large $3$};
      
      \path (1) edge   node {$\mathsf{Cut}$} (2)
            (2) edge   node {$\mathsf{JumpX}$} (0)
            (1) edge   node[below left] {$\mathsf{JumpX}$} (3)
            (3) edge   node[below right] {$\mathsf{Cut}$} (0)
            
            (0) edge[loop right]   node[below right] {$\top$} (0)
            (3) edge[loop below]   node[below right] {$\neg \mathsf{JumpX}$} (3)
            (2) edge[loop above]   node[above] {$\neg \mathsf{Cut}$} (2)
            (1) edge[loop right]   node[right] {$\neg \mathsf{Cut} \land \neg \mathsf{JumpX}$} (1)
            ;
            
    \end{tikzpicture}
    \caption{Inference graph of P2. The edge label $\mathsf{JumpX}$ stands for any of jump action $\mathsf{JumpN, JumpE, JumpS}$ and $\mathsf{JumpW}$. }
        \label{fig:igraph}
\end{figure}
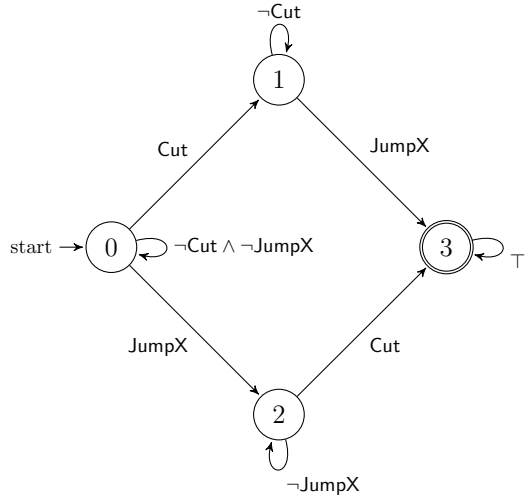

The edges of game graph follow from Definition~\ref{def:game-on-graph} and the game product construction is defined in Appendix~\ref{appendix:ltl}. Given the game graph, the edges of hypergame graph follow from Definition~\ref{def:hgame-on-graph}. A game or hypergame state is marked as a final state whenever \texttt{q} is a final state in the specification DFA. Figure~\ref{fig:dfa} shows the DFAs corresponding to scLTL formulas in Equations~\ref{eq:varphi1} and \ref{eq:varphi2}. In the figure, the final states of DFA are shown with two concentric circles. The inference graph, which captures the evolution of perception of P2, is shown in Figure~\ref{fig:igraph}. The mapping of states of the inference graph to P2's perception of P1's action set is given as follows:
\begin{align*}
    & \mathtt{0: N, E, S, W}, \\
    & \mathtt{1: N, E, S, W, Cut}, \\
    & \mathtt{2: N, E, S, W, JumpN, JumpE, JumpS, JumpW}, \\
    & \mathtt{3: N, E, S, W, JumpN, JumpE, JumpS, JumpW, Cut},
\end{align*}

P2's inference graph state transitions from state $0 \rightarrow 1$ when P1 uses \texttt{Cut} action, from state $0 \rightarrow 2$ when P1 uses any of the jump actions and from $1 \rightarrow 3$ and $2 \rightarrow 3$ when P1 uses any of the jump actions and cut action, respectively. It is noted that the hypergame states in which P2's inference graph state is $3$ corresponds to P2 having complete, symmetric information. That is, if the inference graph had only state $3$ in it, the resulting hypergame graph would coincide with the game with perfect information.

\begin{table}[t!] 
\renewcommand\arraystretch{1.65}
\begin{tabularx}{\columnwidth}{@{}c|c|c|c|c|c|c|c@{}}
\toprule
&$|V|$&$|E|$&$|F|$&\begin{tabular}[c]{@{}c@{}}$|\mathsf{DSWin}_1|$ or \\ $|\dawin_1|$\end{tabular}&\begin{tabular}[c]{@{}c@{}}$|\mathsf{DSWin}_1\downharpoonright_S|$ or \\ $|\dawin_1\downharpoonright_S|$\end{tabular}&$\win_2$&$\mathsf{VoD}$\\
\midrule
SW($\game$) & 4880 & 11449 & 1686 & - & 4724 & 156 & - \\
DSW($\hgame$) & 6965 & 16372 & 2238 & 6734 & 4724 & 156 & 0 \\
DASW($\hgame$)    & 6965 & 16372 & 2238 & \textbf{6947} & \textbf{4868} & \textbf{12} & \textbf{0.9230}\\
\bottomrule
\end{tabularx}
\caption{Comparison of deceptive and non-deceptive winning states under sure and almost-sure winning condition for P1's objective $\varphi_2 = ((\neg\mathsf{FLAG}_2 \land \neg \mathsf{collide}) \until a) \land (\mathsf{collide} \until \mathsf{FLAG}_2)$.}
\label{tbl:results-2}
\end{table}

The result of applying our algorithms on the game and hypergame graph for objective $\varphi_1 = \Eventually \mathsf{FLAG}_1 \land \Eventually \mathsf{FLAG}_2$ is tabulated in Table~\ref{tbl:results-1} and that for objective $\varphi_2 = ((\neg\mathsf{FLAG}_2 \land \neg \mathsf{collide}) \until \mathsf{FLAG}_1) \land (\mathsf{collide} \until \mathsf{FLAG}_2)$ is tabulated in Table~\ref{tbl:results-2}. As expected, for both objectives we observe that the number of deceptive sure winning states is equal to the number of (non-deceptive) sure winning states.

However, under the deceptive \emph{almost}-sure winning condition, we observe that P1 can win from $9395$ out of $9423$ hypergame states. That is, P1 has a deceptive almost-sure winning strategy from $6370$ out of $6388$ game states, which is $6370 - 6133 = 237$ more states than the case when deception is \emph{not} used. This results in $\mathsf{VoD} = 0.9294$. Similarly, for the second objective, where P1 has must capture flags in certain order and ensure that certain safety constraints are also satisfied, we observe that P1 can win from $6947$ out of $6965$ hypergame states. That is, she has a deceptive almost-sure winning strategy from $4868$ out of $4880$ game states which is $4868 - 4724 = 144$ more states than the number of states when deceptive mechanism is not used, thereby, resulting in $\mathsf{VoD} = 0.9230$.

\section{Conclusion}
    \label{sec:conclusion}

In this paper, we have introduced a dynamic hypergame on graph model to represent a game with one-sided incomplete information of action sets. For this class of games, we introduced the notions of deceptive sure and almost-sure winning strategies of P1 and presented algorithms to synthesize them. 
We established two important results regarding the benefit of using deception. First, the use of action deception provides no benefit to P1 when the game is analyzed using sure winning condition, that is, when P1 can ensure to deceptively reach a set of goal states in a finite number of steps. This is because when players use deterministic strategies, P1 cannot be certain when P2 will make a mistake due to his misperception of P1's action set. Second, the use of action deception might be beneficial to P1 when the game is analyzed using almost-sure winning condition, that is, when P1 can ensure to deceptively reach a set of goal states with probability one, with an undetermined number of steps.
This is because when players use randomized strategies,   P2 is ensured to make mistakes with a positive probability. By cleverly designing the deceptive almost-sure strategy, we showed that P1 can be sure that P2 will \emph{almost-surely} make a mistake. 

This work opens several interesting directions for future research. First, our work, which considers one-shot games, can be extended to consider repeated games. The challenge in this extension lies in modeling the effect of revealing an action on future interactions and payoffs in the repeated interactions. Another natural extension is to consider games with stochastic dynamics in which P2 has perfect observation of state history, but not of the action history. In these games, P1 might be able to use private actions without actually revealing them because P2 may attribute her noisy observations of a certain state transition to a P1's action which is which is known to him. Another extension is to  investigate the applications of action deception in   security domain where the hidden actions can be a hidden security countermeasure to adversarial  attackers.

\section*{Acknowledgement}
This material is based upon work in part supported by the Defense Advanced Research Projects Agency (DARPA) under Agreement No. HR00111990015 and in part sponsored by the Army Research Office and Army Research Laboratory (ARL) and was accomplished under Grant Number W911NF-21-1-0114. The views and conclusions contained in this document are those of the authors and should not be interpreted as representing the official policies, either expressed or implied, of the Army Research Office, Army Research Laboratory (ARL) or the U.S. Government. The U.S. Government is authorized to reproduce and distribute reprints for Government purposes notwithstanding any copyright notation herein.

\appendix
\section{Syntactically Co-safe Linear Temporal Logic}
    \label{appendix:ltl}

Syntactically Co-safe Linear Temporal Logic (scLTL) is a subclass of Linear Temporal Logic (LTL) which can be used to represent complex and temporally extended co-safety objectives \cite{kupferman2001model}. Intuitively, a co-safety objective means that something `good' will eventually happen. Formally, an scLTL formula is defined inductively as follows:
\[
    \varphi := \top \mid \bot \mid p \mid \neg p \mid \varphi_1 \land \varphi_2 \mid \bigcirc \varphi \mid \varphi_1 {\until} \varphi_2,
\]
where $\top$ and $\bot$ are universally true and false, respectively, $p \in \calAP$ is an atomic proposition, and $\bigcirc$ is a temporal operator called the ``next'' operator. $\bigcirc \varphi$ is evaluated to be true if the formula $\varphi$ becomes true at the next time step. $\until$ is a temporal operator called the ``until'' operator. The formula $\varphi_1\until \varphi_2$ is true given that $\varphi_2$ will be true in some future time steps, and before that $\varphi_1$ holds true for every time step. The operator $\Eventually$ (read as eventually) is defined using the operator $\until$ as follows: $\Eventually \varphi = \top \until \varphi$. The formula $\Eventually \varphi$ is true if $\varphi$ becomes true in some future time.  Given an scLTL formula $\varphi$ and a word $w\in \Sigma^\omega$, if the word $w$ satisfies the formula $\varphi$, then we denote $w\models \varphi$. For details about the syntax and semantics of scLTL, the readers are referred to \cite{pnueli1989synthesis,kupferman2001model}.

An scLTL formula contains only $\Eventually$ and $\until$ temporal operators when written in a positive normal form (\ie the negation operator $\neg$ appears only in front of atomic propositions). The unique property of scLTL formulas is that a word satisfying an scLTL formula $\varphi$ only needs to have a \emph{good prefix}. That is, given a good prefix $w \in \Sigma^\ast$, the word $ww' \models \varphi$ for any $w'\in \Sigma^\omega$. The set of good prefixes can be compactly represented as the language accepted by a \ac{dfa} defined as follows:
\begin{definition}[Deterministic Finite Automaton]
Given an scLTL formula $\varphi$, the set of good prefixes of words corresponding to $\varphi$ is accepted by a \ac{dfa} \[{\cal A} = \langle Q, \Sigma, \delta, \iota, Q_F \rangle\] with the following components:
\begin{itemize}
    \item $Q$ is a finite set of states.
    \item $\Sigma = 2^{\calAP}$ is a finite set of symbols.
    \item $\delta \colon Q \times \Sigma \to Q$ is a deterministic transition function.
    \item $\iota \in Q$ is a unique initial state.
    \item $Q_F \subseteq Q$ is a set of final states.
\end{itemize}
\end{definition}

For an input word $w =w_0 w_1 \ldots \in \Sigma^\omega$, the \ac{dfa} generates a sequence of states $q_0q_1\ldots $ such that $q_0=\iota$ and $q_{i+1}= \delta(q_{i},w_{i})$ for any $i \geq 0$. The word $w$ is accepted by the \ac{dfa} if and only if there exists $k \ge 0$ such that $q_k \in Q_F$. The set of words accepted by the \ac{dfa} $\mathcal{A}$ is called \emph{its language}. We assume that the \ac{dfa} is complete. That is, for every state-action pair $(q, w_i)$, $\delta(q, w_i)$ is well-defined. An incomplete DFA can be made complete by adding a sink state $q_{\mathsf{sink}}$ such that $\delta(q_{\mathsf{sink}}, w_i)= q_{\mathsf{sink}}$ and directing all undefined transitions to the sink state $q_{\mathsf{sink}}$.

Given a specification \ac{dfa} corresponding to an scLTL specification $\varphi$, a reachability game with complete, symmetric information in Definition~\ref{def:game-on-graph} is constructed as a product of a game transition system and the \ac{dfa} \cite{baier2008principles}. A game transition system captures the dynamics of the interaction between P1 and P2, and is formally defined as follows:

\begin{definition}[Game Transition System]
    \label{def:game-tsys}
    
    A game transition system capturing the dynamics of the interaction between P1 and P2 is defined as the tuple,
    \[
        GTS = \langle \hat{S}, A_1 \cup A_2, \Delta, \mathcal{AP}, L \rangle
    \]
    with the following components:
    \begin{itemize}
        \item $\hat{S} = \hat{S}_1 \cup \hat{S}_2$ is the set of game transition system states partitioned into P1 and P2 states,
        
        \item $A_1 \cup A_2$ are P1 and P2 actions,
        
        \item $\Delta: \hat{S} \times (A_1 \cup A_2) \rightarrow \hat{S}$ is a deterministic transition function,
        
        \item $\cal AP$ is the set of atomic propositions,
        
        \item $L: \hat{S} \rightarrow 2^{\cal AP}$ is a labeling function which maps every state in $\hat{S}$ to a set of atomic proposition which are true in that state.
    \end{itemize}
\end{definition}

Given a game transition system $GTS$ and a \ac{dfa} $\mathcal{A} = \langle Q, \Sigma=2^{\cal AP}, \delta, \iota, Q_F \rangle$, the components of the reachability game  $\game(A_1) = \langle S, A_1 \cup A_2, T, F\rangle$ are defined as follows: 
\begin{itemize}
    \item $S = \hat{S} \times Q$, where the P1 and P2 states are defined as $S_1 = \hat{S}_1 \times Q$ and $S_2 = \hat{S}_2 \times Q$,
    
    \item Given two states $s = (\hat{s}, q)$, $s' = (\hat{s}', q')$ and an action $a \in A_1 \cup A_2$, we have $T(s, a) = s'$ if and only if $\Delta(\hat{s}, a) = \hat{s}'$ and $\delta(q, L(\hat{s}') = q'$,
    
    \item $F = \hat{S} \times Q_F$ is the set of final states in $\game(A_1)$. 
\end{itemize}
 
Intuitively, the product operation is defined such that a game-run visiting one of the final states in $F$ respects the dynamics of the interaction between P1 and P2 and also satisfies P1's scLTL objective $\varphi$.

    


\vskip 0.2in
\bibliography{action-deception}
\bibliographystyle{template/theapa}

\end{document}